\documentclass[11pt]{article}
\usepackage{amsmath,amssymb,amsfonts}
\usepackage{amsthm}
\usepackage{array}
\usepackage{graphicx}
\usepackage{xcolor}
\usepackage{enumitem}
\usepackage{url}
\usepackage{hyperref}
\usepackage{aliascnt}
\usepackage{cleveref}

\theoremstyle{plain}
\newtheorem{theorem}{Theorem}

\newaliascnt{lemma}{theorem}
\newtheorem{lemma}[lemma]{Lemma}
\aliascntresetthe{lemma}

\newaliascnt{proposition}{theorem}
\newtheorem{proposition}[proposition]{Proposition}
\aliascntresetthe{proposition}

\newaliascnt{corollary}{theorem}
\newtheorem{corollary}[corollary]{Corollary}
\aliascntresetthe{corollary}

\theoremstyle{definition}
\newaliascnt{definition}{theorem}
\newtheorem{definition}[definition]{Definition}
\aliascntresetthe{definition}

\newaliascnt{remark}{theorem}
\newtheorem{remark}[remark]{Remark}
\aliascntresetthe{remark}

\newaliascnt{example}{theorem}

\aliascntresetthe{example}

\crefname{theorem}{Theorem}{Theorems}
\Crefname{theorem}{Theorem}{Theorems}
\crefname{lemma}{Lemma}{Lemmas}
\Crefname{lemma}{Lemma}{Lemmas}
\crefname{corollary}{Corollary}{Corollaries}
\Crefname{corollary}{Corollary}{Corollaries}
\crefname{proposition}{Proposition}{Propositions}
\Crefname{proposition}{Proposition}{Propositions}
\crefname{definition}{Definition}{Definitions}
\Crefname{definition}{Definition}{Definitions}
\crefname{remark}{Remark}{Remarks}
\Crefname{remark}{Remark}{Remarks}
\crefname{example}{Example}{Examples}
\Crefname{example}{Example}{Examples}
\crefname{table}{Table}{Tables}
\Crefname{table}{Table}{Tables}
\crefname{figure}{Figure}{Figures}
\Crefname{figure}{Figure}{Figures}
\crefname{section}{Section}{Sections}
\Crefname{section}{Section}{Sections}
\crefname{equation}{}{}
\Crefname{equation}{Equation}{Equations}

\DeclareMathOperator{\rank}{rank}
\DeclareMathOperator{\Per}{Per}
\DeclareMathOperator{\wt}{wt}
\DeclareMathOperator{\PG}{PG}
\newcommand{\gen}[1]{\langle #1 \rangle}

\begin{document}

\title{Caps and LCD Codes with Large Dimensions}

\author{Keita Ishizuka\thanks{Information Technology R\&D Center, Mitsubishi Electric Corporation, Japan, \texttt{keitaishizuka1994@gmail.com}}
\and
Yuhi Kamio\thanks{
Graduate School of Mathematical Sciences, University of Tokyo
\texttt{emirp13@g.ecc.u-tokyo.ac.jp}}
}
\date{}

\maketitle

\begin{abstract}
We establish a connection between linear complementary dual (LCD) codes and caps in projective space. Using this framework and the structure theory of maximal binary caps, we derive nonexistence theorems for binary LCD codes with minimum distance at least \(4\), providing computation-free proofs that were previously obtained only through exhaustive search. As an application, we completely determine the optimal minimum distances for codimensions \(7\) and \(8\) for the first time. The framework applies to ternary LCD codes and quaternary Hermitian LCD codes, where it reduces the problem to the structure of large caps. Using known results on large caps, we determine the optimal minimum distances of ternary LCD codes of codimension \(6\) for all lengths, and we explain why the phenomena behind the binary nonexistence theorems, such as the dependence on the parity of the length, have no counterpart for these codes.
\end{abstract}


\section{Introduction}

Linear complementary dual (LCD) codes are codes that intersect their dual codes trivially. Introduced by Massey \cite{massey1992linear} for the two-user binary adder channel, LCD codes have found important applications in cryptography, particularly for protection against side-channel attacks and fault injection attacks \cite{carlet2016complementary}.
A breakthrough result by Carlet et al. \cite{carlet2018equivalence} showed that any code over $\mathbb{F}_q$ is equivalent to some Euclidean LCD code for $q \geq 4$, and that any code over $\mathbb{F}_{q^2}$ is equivalent to some Hermitian LCD code for $q \geq 3$. This motivates the study of binary and ternary Euclidean LCD codes and quaternary Hermitian LCD codes, as these are the only cases where the LCD property imposes genuine constraints on the code structure. For binary codes, determining the largest minimum distance $d_2^E(n,k)$ among all LCD $[n,k]$ codes is a fundamental problem.

For arbitrary $n$ and various fixed $k$, the values $d_2^E(n,k)$ have been previously determined: $d_2^E(n,1)$ through $d_2^E(n,5)$ for small dimensions, and $d_2^E(n,n-i)$ for $i \in \{1,2,3,4,5\}$ for small codimensions (see \cite{characterizations2024} for a comprehensive list). The determination of $d_2^E(n,n-m)$ for larger codimensions $m$ has remained a challenge.
Recently, the first author and collaborators \cite{characterizations2024} made progress in understanding LCD codes of large dimensions. Their key insight was to use Hamming codes---whose parity-check matrices ensure any $2$ columns are linearly independent---to analyze when $d_2^E(n,n-m) \geq 3$. Through this approach, they characterized the transition to $d_2^E(n,n-m) = 2$, and determined $d_2^E(n,n-6)$ completely and $d_2^E(n,n-7)$ with some exceptions. However, their Hamming code framework cannot address when $d_2^E(n,n-m) \geq 4$, as this requires $3$-independence rather than $2$-independence.

This limitation has significant consequences. For $m=6$, computational results \cite{araya2020minimum,bouyuklieva2021optimal} revealed that $d_2^E(n,n-6)$ exhibits an alternating pattern between $3$ and $4$ based on the parity of $n$ (for $17 \leq n \leq 26$). While exhaustive computational verification confirmed that no LCD $[n,n-6,4]$ codes exist for odd $n$ in this range, the lack of a theoretical framework for $d \geq 4$ forces reliance on such brute-force methods. This computational bottleneck makes it infeasible to extend these results to larger codimensions, leaving the determination of $d_2^E(n,n-m)$ for $m \geq 7$ open.
This motivates us to uncover the hidden structure behind the alternating pattern, thereby establishing nonexistence results and extending the determination of $d_2^E(n,n-m)$ to larger codimensions where exhaustive computation is infeasible.

\paragraph{Our Contribution.}

This paper extends the framework from Hamming codes ($2$-independent sets) to caps ($3$-independent sets). By analyzing the structure of large caps, we uncover the geometry behind the alternating pattern and establish nonexistence theorems for LCD codes with $d \geq 4$.

Our key contributions are:
\begin{enumerate}
    \item \emph{LCD characterization via caps}: For a cap $S \subseteq \mathbb{F}_q^m$ with $q \in \{2,3,4\}$, the code $C_S$ is LCD (Hermitian LCD if $q = 4$) if and only if its Gram matrix $U_S = H_S \overline{H_S}^T$ is nonsingular, where $H_S$ is the matrix whose columns are the elements of $S$. This transforms the algebraic LCD problem into a geometric question about point configurations.
    \item \emph{Structure theorem for large caps}: Using Bruen and Wehlau's theory \cite{bruen1999long}, we establish a sufficient condition for caps to lie in hyperplane complements (\cref{thm:main}). This result is of independent interest in cap theory and provides the key geometric constraint.
    \item \emph{Nonexistence theorems}: Combining the above results, we prove that, for every $m \geq 4$ and every sufficiently large $n$, any LCD $[n,n-m,d]$ code with $d \geq 4$ must satisfy $n \equiv m \pmod{2}$ and $n \leq 2^{m-1} - m$ (\cref{thm:nonexist} and~\cref{thm:nonexist_small}).
\end{enumerate}

These theorems provide computation-free proofs for the nonexistence results underlying the alternating pattern in $d_2^E(n,n-m)$, which previously required exhaustive search. As a consequence, we completely determine $d_2^E(n,n-7)$ and $d_2^E(n,n-8)$ for all $n$, resolving the open cases that were computationally infeasible.

\paragraph{Organization.}
The rest of this paper is organized as follows: \cref{sec:prelim} reviews the necessary background on LCD codes over $\mathbb{F}_q$, $q \in \{2,3,4\}$, and caps, and introduces the Gram matrix of a cap. \cref{sec:gram} turns to binary caps: it recalls the classification of large binary caps, relates the period of a cap to the rank of its Gram matrix, and proves our structure theorem for large caps. \cref{sec:det} applies these results to completely determine $d_2^E(n,n-m)$ for codimensions $6$, $7$, and $8$. \cref{sec:nonbinary} treats ternary and quaternary Hermitian LCD codes: using known results on large caps, we settle cases left open in~\cite{characterizations2024}, and we explain why no analogue of our theorems holds for these codes. \cref{sec:conc} concludes with open problems.

\section{Preliminaries}\label{sec:prelim}

\subsection{Linear Codes and LCD Codes}

Throughout this paper, $q \in \{2,3,4\}$ and $\mathbb{F}_q$ denotes the finite field of order $q$.
As explained in the introduction, these are the cases in which the LCD property imposes genuine constraints~\cite{carlet2018equivalence}, and the main results of this paper concern the case $q = 2$.
An $[n,k]$ code $C$ over $\mathbb{F}_q$ is a $k$-dimensional subspace of $\mathbb{F}_q^n$ not equal to $\{0\}$; codes over $\mathbb{F}_2$, $\mathbb{F}_3$ and $\mathbb{F}_4$ are called \emph{binary}, \emph{ternary} and \emph{quaternary}, respectively.
All codes in this paper are linear, and we omit the term ``linear'' when referring to codes.
For a vector $v = (v_1, \ldots, v_n) \in \mathbb{F}_q^n$, the \emph{Hamming weight} is defined as $\wt(v) = |\{i : v_i \neq 0\}|$.
A code $C$ is called \emph{even} if $\wt(c) \equiv 0 \pmod{2}$ for all $c \in C$.
The \emph{Hamming distance} between vectors $u, v \in \mathbb{F}_q^n$ is $d(u,v) = \wt(u-v)$.
The \emph{minimum distance} (or \emph{minimum weight}) of a code $C$ is $d(C) = \min\{\wt(c) : c \in C \setminus \{0\}\}$.
An $[n,k]$ code with minimum distance $d$ is denoted as an $[n,k,d]$ code.

For $x \in \mathbb{F}_q$, we put $\overline{x} = x$ if $q \in \{2,3\}$ and $\overline{x} = x^2$ if $q = 4$, where we write $\mathbb{F}_4 = \{0,1,\omega,\omega^2\}$; thus $x \mapsto \overline{x}$ is the unique nontrivial automorphism of $\mathbb{F}_4$, and $\overline{\overline{x}} = x$ in all cases.
For a vector $v$ or a matrix $A$ over $\mathbb{F}_q$, we denote by $\overline{v}$ or $\overline{A}$ the vector or matrix obtained by applying $x \mapsto \overline{x}$ to all entries.
For vectors $x = (x_1, \ldots, x_n)$ and $y = (y_1, \ldots, y_n)$ in $\mathbb{F}_q^n$, we define the inner product $\langle x, y \rangle = \sum_{i=1}^n x_i \overline{y_i}$; this is the standard (Euclidean) inner product if $q \in \{2,3\}$ and the Hermitian inner product if $q = 4$.
The dual code $C^\perp$ of $C$ is defined as $C^\perp = \{x \in \mathbb{F}_q^n : \langle x, c \rangle = 0 \text{ for all } c \in C\}$; for $q = 4$ this is the Hermitian dual code.
The \emph{hull} of a code $C$ is defined as $\mathrm{Hull}(C) = C \cap C^\perp$.
A code $C$ is called \emph{linear complementary dual (LCD)} if $\mathrm{Hull}(C) = \{0\}$, \emph{self-orthogonal} if $C \subseteq C^\perp$ (equivalently, $\dim \mathrm{Hull}(C) = k$), and \emph{self-dual} if $C = C^\perp$.
For $q = 4$, LCD codes and self-orthogonal codes in this sense are usually called \emph{Hermitian LCD} codes and \emph{Hermitian self-orthogonal} codes, and we also use these terms.
Note that $\mathrm{Hull}(C^\perp) = \mathrm{Hull}(C)$, so a code is LCD if and only if its dual code is LCD.
Following~\cite{characterizations2024}, we denote by $d_2^E(n,k)$, $d_3^E(n,k)$ and $d_4^H(n,k)$ the largest minimum distance among all binary LCD, ternary LCD and quaternary Hermitian LCD $[n,k]$ codes, respectively.
For a code $C$ with generator matrix $G$, the dimension of the hull can be computed from the Gram matrix $G\overline{G}^T$, which is $GG^T$ if $q \in \{2,3\}$.

\begin{lemma}[{\cite[Propositions~3.1 and~3.2]{guenda2018constructions}}]\label{lem:hull-dim}
Let $C$ be an $[n,k]$ code over $\mathbb{F}_q$ with generator matrix $G$. Then
\[
\dim \mathrm{Hull}(C) = k - \rank(G\overline{G}^T).
\]
In particular, $C$ is LCD if and only if $G\overline{G}^T$ is nonsingular, and self-orthogonal if and only if $G\overline{G}^T=0$.
\end{lemma}

The following lemma is specific to the binary case.

\begin{lemma}[{\cite[Theorem 5]{carlet2018hull}}]\label{lem:even-lcd}
Binary even LCD codes must have even dimensions.
\end{lemma}

\subsection{Caps}

Throughout this paper, vectors are regarded as row vectors.
For a subset $S \subseteq \mathbb{F}_q^m$, we denote by $\langle S \rangle$ the subspace spanned by $S$, and we write $\langle v \rangle = \langle \{v\} \rangle$ for $v \in \mathbb{F}_q^m$.
We denote by $e_1, e_2, \ldots, e_m$ the standard basis vectors of $\mathbb{F}_q^m$, and by $\mathbf{1}$ the all-one vector.

A subset $S \subseteq \mathbb{F}_q^m$ is called a \emph{cap} if every subset $T \subseteq S$ with $\lvert T \rvert \leq 3$ is linearly independent.
In particular, a cap contains no two linearly dependent vectors, and hence it can be viewed as a set of points in the projective space $\PG(m-1, q)$; the cap condition says that no three of these points are collinear.
A cap $S$ is called \emph{maximal} if it is not properly contained in any other cap.

\subsection{Connection Between Caps and Codes}\label{sec:cap-lcd}

For any subset $S = \{s_1, s_2, \ldots, s_n\} \subseteq \mathbb{F}_q^m$, we define the matrix $H_S = (s_1^T, s_2^T, \ldots, s_n^T) \in \mathbb{F}_q^{m \times n}$ whose columns are the elements of $S$.
Note that $H_S$ is determined up to column permutations. When $\rank(H_S) = m$, we can use $H_S$ as a parity-check matrix to define an $[n, n-m]$ code $C_S = \{x \in \mathbb{F}_q^n : H_S x^T = 0\}$ over $\mathbb{F}_q$, which is determined up to coordinate permutations.
The following folklore results establish the fundamental connection between caps and codes with minimum distance at least $4$, whose proofs are included for completeness.

\begin{proposition}\label{prop:lcd-cap}
Every $[n,k,d]$ code $C$ over $\mathbb{F}_q$ with $d \geq 4$ can be expressed as $C_S$ for some cap $S \subseteq \mathbb{F}_q^{n-k}$ with $\lvert S \rvert = n$ and $\langle S \rangle = \mathbb{F}_q^{n-k}$, up to coordinate permutations.
\end{proposition}

\begin{proof}
Let $C$ be an $[n,k,d]$ code with $d \geq 4$ and parity-check matrix $H \in \mathbb{F}_q^{(n-k) \times n}$. Let $S$ be the set of columns of $H$. Since $d \geq 4$, any three columns of $H$ are linearly independent, which means $S$ is a cap with $\lvert S \rvert = n$. Since $H$ has rank $n-k$, we have $\langle S \rangle = \mathbb{F}_q^{n-k}$. Moreover, by construction, $C$ equals $C_S$ up to coordinate permutations.
\end{proof}

\begin{proposition}\label{prop:cap-to-code}
Let $S \subseteq \mathbb{F}_q^m$ be a cap with $\lvert S \rvert = n$ and $\langle S \rangle = \mathbb{F}_q^m$. Then $\rank(H_S) = m$ and $C_S$ is an $[n, n-m]$ code with minimum distance at least $4$.
\end{proposition}

\begin{proof}
Since $\langle S \rangle = \mathbb{F}_q^m$, we have $\rank(H_S) = m$. Since $S$ is a cap, any three columns of $H_S$ are linearly independent, which implies that the minimum distance of $C_S$ is at least $4$.
\end{proof}

\subsection{The Gram Matrix}\label{sec:gram-def}

The following definition turns the LCD property of $C_S$ into a property of the cap $S$.

\begin{definition}
For a subset $S \subseteq \mathbb{F}_q^m$, define the matrix
\[
U_S := H_S \overline{H_S}^T \in \mathbb{F}_q^{m \times m},
\]
which is $H_S H_S^T$ if $q \in \{2,3\}$.
We call this matrix the Gram matrix of $S$.
\end{definition}
As $H_S$ is determined up to column permutations, the matrix $U_S$ is uniquely determined by $S$ (we put $U_\emptyset = 0$).
When $\langle S \rangle = \mathbb{F}_q^m$, the matrix $\overline{H_S}$ is a generator matrix of $C_S^\perp$, and $\overline{U_S}$ is exactly the Gram matrix $G\overline{G}^T$ in \cref{lem:hull-dim} for the generator matrix $G = \overline{H_S}$ of $C_S^\perp$; in particular, $\rank(U_S) = \rank(\overline{U_S}) = \dim C_S^\perp - \dim \mathrm{Hull}(C_S)$.
If $A, B \subseteq \mathbb{F}_q^m$ are disjoint, then $H_{A \sqcup B} = (H_A \;\, H_B)$ up to column permutations, and hence
\begin{equation}\label{eq:additivity}
U_{A \sqcup B} = U_A + U_B.
\end{equation}
Moreover, the rank of $U_S$ depends only on the projective equivalence class of $S$.
Indeed, if $A$ is an invertible $m \times m$ matrix over $\mathbb{F}_q$ and $S' = \{sA : s \in S\}$, then $H_{S'} = A^T H_S$ and hence
\begin{equation}\label{eq:equivalence}
U_{S'} = A^T U_S \overline{A};
\end{equation}
replacing an element $s \in S$ by $\lambda s$ with $\lambda \in \mathbb{F}_q \setminus \{0\}$ does not change $U_S$, since $\lambda\overline{\lambda} = 1$ for $q \in \{2,3,4\}$; and if $S' = \{\overline{s} : s \in S\}$, then $U_{S'} = \overline{U_S}$.

\begin{proposition}\label{prop:lcd-nonsingular}
Let $S \subseteq \mathbb{F}_q^m$ be a subset with $\langle S \rangle = \mathbb{F}_q^m$.
Then the code $C_S$ is LCD if and only if $U_S$ is nonsingular.
\end{proposition}

\begin{proof}
Since $C_S$ has parity-check matrix $H_S$, its dual $C_S^\perp$ has generator matrix $\overline{H_S}$. By \cref{lem:hull-dim}, $C_S^\perp$ is LCD if and only if $\overline{H_S} H_S^T = \overline{U_S}$ is nonsingular, that is, if and only if $U_S$ is nonsingular. The result follows since a code is LCD if and only if its dual is LCD.
\end{proof}

The following bound on the rank of $U_S$ is used throughout.

\begin{lemma}[Rank Bound]\label{lem:rank-bound}
Let $S \subseteq \mathbb{F}_q^m$. Then $\rank(U_S) \leq \rank(H_S) \leq \lvert S \rvert$.
\end{lemma}

\begin{proof}
Since $H_S$ has $\lvert S \rvert$ columns, we have $\rank(U_S) = \rank(H_S \overline{H_S}^T) \leq \rank(H_S) \leq \lvert S \rvert$.
\end{proof}

\section{The Gram Matrix and Large Binary Caps}\label{sec:gram}

In this section and the next, we consider the binary case $q = 2$.
We first recall the classification of large binary caps by Bruen and Wehlau, and then establish key properties of the matrix $U_S$, particularly focusing on how the period structure affects its rank.
These structural results enable us to prove that sufficiently large caps must be contained in a hyperplane complement.

\subsection{Large Binary Caps}\label{sec:binary-caps}

For $v, w \in \mathbb{F}_2^m$ and $S \subseteq \mathbb{F}_2^m$, we write $\langle v \rangle^\perp = \{u \in \mathbb{F}_2^m : \langle u, v \rangle = 0\}$ and $S + w = \{s + w : s \in S\}$.
A \emph{hyperplane} of $\mathbb{F}_2^m$ is a subspace of the form $\langle v \rangle^\perp$ with $v \neq 0$.
A cap $S \subseteq \mathbb{F}_2^m$ is called \emph{large} if $\lvert S \rvert \geq 2^{m-2} + 1$.

\begin{definition}
For a subset $S \subseteq \mathbb{F}_2^m$, the \emph{period} of $S$ is defined as
\[
\Per(S) = \{v \in \mathbb{F}_2^m : s + v \in S \text{ for all } s \in S\}.
\]
\end{definition}

Note that $\Per(S)$ is always a subspace of $\mathbb{F}_2^m$, and hence $|\Per(S)|$ is a power of $2$.
The following deep results from Bruen and Wehlau \cite{bruen1999long} characterize large maximal caps.

\begin{theorem}[{\cite[Corollary 3.14]{bruen1999long}; see also \cite{DavydovTombak}}]\label{thm:bw}
Let $S \subseteq \mathbb{F}_2^m$ be a maximal cap with $\lvert S \rvert \geq 2^{m-2} + 1$ and $m \geq 3$. Then there exists $j \in \{0,1,\ldots,m-4\} \cup \{m-2\}$ such that $\lvert S \rvert = 2^{m-2} + 2^j$. Moreover, the following hold:
\begin{enumerate}[label=\textup{(\arabic*)}]
\item If $j \leq m-4$, then $\lvert \Per(S) \rvert = 2^j$.
\item If $j = m-2$, then $S = \mathbb{F}_2^m \setminus \langle v \rangle^\perp$ for some nonzero $v \in \mathbb{F}_2^m$; in this case $\Per(S) = \langle v \rangle^\perp$, which has $2^{m-1}$ elements.
\end{enumerate}
\end{theorem}

\begin{remark}
Bruen and Wehlau~\cite{bruen1999long} work with projective space notation $\PG(n,2)$, which corresponds to $\mathbb{F}_2^{n+1}$ in our vector space notation. Therefore, their dimension parameter $n$ corresponds to our $m-1$. We have adapted their results to match our notation for consistency.
\end{remark}

\subsection{Periods and the Gram Matrix}

We now establish bounds on the rank of $U_S$ based on the period structure of $S$. These bounds are crucial for understanding when LCD codes can exist.

\begin{lemma}[Rank Bounds for Binary Caps]\label{lem:rank-bounds}
Let $S \subseteq \mathbb{F}_2^m$. Then the following hold:
\begin{enumerate}[label=\textup{(\arabic*)}]
\item If $w \in \Per(S)$, $w \neq 0$,
then the linear map $U_S$ sends $\gen{w}^{\perp}$ into $\gen{w}$. In particular, $\rank(U_S) \leq 2$ if $\Per(S)\neq \{0\}$.
\item If $|\Per(S)| \geq 8$, then $U_S = 0$.
\item Assume that $S$ is a cap, $m \geq 5$, $\lvert S \rvert=2^{m-2}+1$, and there exists a hyperplane $H$ such that $\lvert S\cap H\rvert=1$. Then $\rank(U_S)\leq 2$.
\end{enumerate}
\end{lemma}

\begin{proof}
(1) Since $w \in \Per(S)$ and $w \neq 0$, the map $s \mapsto s + w$ is an involution on $S$ without fixed points. Hence $S$ is partitioned into pairs $\{s, s+w\}$, and choosing one element from each pair gives a subset $T \subseteq S$ with $S = T \sqcup (T + w)$.
By the definition of $U_S$, we have $U_S = \sum_{s \in S} s^T s$, where $s^T s$ is the $m \times m$ matrix whose $(i,j)$ entry is the product of the $i$th and $j$th coordinates of $s$.
Hence, for any $u \in \langle w \rangle^\perp$, we have
\begin{align*}
uU_S &= \sum_{v \in S} \langle u, v \rangle v \\
&= \sum_{v \in T} \langle u, v \rangle v + \sum_{v \in T} \langle u, v + w \rangle (v + w) \\
&= \sum_{v \in T} \langle u, v \rangle w \in \langle w \rangle.
\end{align*}
This proves the first assertion. For the second, choose $v \notin \langle w \rangle^\perp$. Then $\operatorname{Im}(U_S) = \langle vU_S \rangle + \langle w \rangle^\perp U_S \subseteq \langle vU_S, w \rangle$, so $\rank(U_S) \leq 2$.

(2) Let $v \in \mathbb{F}_2^m$ be arbitrary. Since $\Per(S) \cap \langle v \rangle^\perp$ has dimension at least $2$, we can choose linearly independent $a, b \in \Per(S) \cap \langle v \rangle^\perp$. From (1), for $w \in \{a, b\}$, $U_S$ maps $\langle w \rangle^\perp$ to $\langle w \rangle$. Therefore, $vU_S \in \langle a \rangle \cap \langle b \rangle = \{0\}$. Since $v$ was arbitrary, $U_S = 0$.

(3) Let $v$ be the unique element of $S\cap H$; note that $v \neq 0$ since $S$ is a cap. Choose a hyperplane $H'$ not containing $v$, and fix it for the rest of the proof; note that $H' \neq H$ since $v \in H$.
Let $S_1=(S\setminus H)\cap H'$ and $S_2=(S\setminus H)\setminus H'$. By definition, $S=\{v\}\sqcup S_1\sqcup S_2$, and in particular $\lvert S_1\rvert +\lvert S_2\rvert=2^{m-2}$.

Since $S$ is a cap, $v+S_2$ is disjoint from $S_1$. On the other hand, $S_1, v+S_2 \subset H'\setminus H$ (for $v+S_2$, note that $v \in H \setminus H'$ and $S_2 \cap (H \cup H') = \emptyset$), and $\lvert H'\setminus H \rvert=2^{m-2}=\lvert S_1\rvert +\lvert S_2\rvert$. Therefore, $H'\setminus H = S_1 \sqcup (v+S_2)$.

Since $\lvert \Per(H'\setminus H) \rvert= \lvert H\cap H' \rvert \geq 8$, we have $U_{H'\setminus H}=0$ by (2). Therefore,
\[U_S=U_S+U_{H'\setminus H}=U_{\{v\}}+U_{S_2}+U_{S_1}+U_{H'\setminus H}= U_{\{v\}}+U_{S_2}+U_{v+S_2}.\]
Since $v \in \Per((v+S_2)\sqcup S_2)$, the linear map $U_{S_2}+U_{v+S_2}$ sends $\gen{v}^{\perp}$ into $\gen{v}$ by (1). The map $U_{\{v\}}$ also sends $\gen{v}^{\perp}$ into $\gen{v}$, so $U_S$ has the same property and we conclude $\rank(U_S)\leq 2$ as in (1).
\end{proof}

\begin{lemma}[Even Code Characterization]\label{lem:even-code}
Let $S \subseteq \mathbb{F}_2^m$ be a cap with $\langle S \rangle = \mathbb{F}_2^m$.
Then $C_S$ is an even code if and only if $S \cap \langle v \rangle^\perp = \emptyset$ for some $v \in \mathbb{F}_2^m$.

\end{lemma}
\begin{proof}
Over $\mathbb{F}_2$, the Hamming weight satisfies $\wt(w) = \langle w, \mathbf{1} \rangle$ for every $w \in \mathbb{F}_2^n$.
Therefore, $C_S$ is even if and only if $\mathbf{1} \in C_S^\perp$.
Since $\langle S \rangle = \mathbb{F}_2^m$, the code $C_S$ is defined and $H_S$ is a generator matrix of $C_S^\perp$. Hence the condition $\mathbf{1} \in C_S^\perp$ holds if and only if $\mathbf{1} = vH_S$ for some $v \in \mathbb{F}_2^m$.
Since the entries of $vH_S$ are $\langle v, s \rangle$ for $s \in S$, this is equivalent to $\langle v, s \rangle = 1$ for every $s \in S$, which means $S \cap \langle v \rangle^\perp = \emptyset$.
\end{proof}

\subsection{Structure Theorem}

\begin{theorem}[Structure of Large Caps]\label{thm:main}
Let $m \geq 7$.
Any cap $S$ with cardinality $\lvert S \rvert \geq 2^{m-2} + 2^{m-4} - \rank(U_S) + 1$ must be contained in a maximal cap of the form $\mathbb{F}_2^m \setminus \langle v \rangle^\perp$ for some nonzero $v \in \mathbb{F}_2^m$.
Furthermore, $\lvert S \rvert \le 2^{m-1} - \rank(U_S)$.
\end{theorem}

\begin{proof}
Let $T$ be a maximal cap containing $S$.
Since $\lvert S \rvert \geq 2^{m-2} + 2^{m-4} - \rank(U_S) + 1$ and $m \geq 7$, we have
\begin{align*}
\lvert T \rvert \geq \lvert S \rvert
&\geq 2^{m-2} + 2^{m-4} - \rank(U_S) + 1 \\
&\geq 2^{m-2} + 2^{m-4} - m + 1 \\
&\geq 2^{m-2} + 2.
\end{align*}

By \cref{thm:bw}, $\lvert T \rvert = 2^{m-2} + 2^j$ for some $j \in \{ 1,2, \ldots, m-4, m-2\}$.
Let $X = T \setminus S$. Since $T = S \sqcup X$ (disjoint union), we have $U_T = U_S + U_X$ by \eqref{eq:additivity}. Since we are working over $\mathbb{F}_2$, this gives $U_S = U_T + U_X$.
By combining this with \cref{lem:rank-bound},
\[
\rank(U_T) + |X| \geq \rank(U_T) + \rank(U_X) \geq \rank(U_S).
\]
Therefore, we have
\begin{align*}
\rank(U_T)
&\geq \rank(U_S) - |X|\\
&= \rank(U_S) - (\lvert T \rvert - \lvert S \rvert)\\
&\geq \rank(U_S) - ((2^{m-2} + 2^j) - (2^{m-2} + 2^{m-4} - \rank(U_S) + 1))\\
&= 2^{m-4} - 2^j + 1.
\end{align*}
Now we analyze the possible values of $j$:

(i) If $j=1$ or $2$, then $\lvert \Per(T) \rvert = 2^j$ by \cref{thm:bw}(1), so $\Per(T) \neq \{0\}$ and $\rank(U_T) \leq 2$ by \cref{lem:rank-bounds}(1).
However, the above inequality gives
\[
2 \ge \rank(U_T) \ge 2^{m-4} - 2^j + 1 \ge 2^{m-4} - 3,
\]
which cannot be satisfied for $m \ge 7$.
Thus, this case is impossible.

(ii) If $j \geq 3$, then $\lvert \Per(T) \rvert = 2^j \geq 8$ by \cref{thm:bw}(1) when $j \leq m-4$, and $\lvert \Per(T) \rvert = 2^{m-1} \geq 8$ by \cref{thm:bw}(2) when $j = m-2$. In both cases $\rank(U_T) = 0$ by \cref{lem:rank-bounds}(2).
Suppose for contradiction that $j \le m-4$.
Then the above inequality gives
\[
0 \ge \rank(U_T) \ge 2^{m-4} - 2^j + 1 \ge 1,
\]
which is absurd.
Therefore, we must have $j = m-2$, which means $T = \mathbb{F}_2^m \setminus \langle v \rangle^\perp$ for some nonzero $v \in \mathbb{F}_2^m$ by \cref{thm:bw}.

Finally, we will show that $\lvert S \rvert \le 2^{m-1} - \rank(U_S)$. 
Since $\Per(T)=\gen{v}^{\perp}$ has dimension $m-1 \geq 6$ by \cref{thm:bw}(2), we have $U_T=0$ by \cref{lem:rank-bounds}(2).
Therefore, we have
\[U_S + U_X = U_T = 0,\]
giving $U_S = U_X$.
Applying this to \cref{lem:rank-bound}, we have $|X| \ge \rank(U_X) = \rank(U_S)$.
Therefore, we have
\[
\lvert S \rvert = \lvert T \rvert - |X| \le 2^{m-1} - \rank(U_S),
\]
as was to be shown.
This completes the proof.

\end{proof}

\section{Determination of $d_2^E(n,n-m)$}\label{sec:det}

We now apply the results of the previous section to completely determine the values of $d_2^E(n,n-m)$ for codimensions $6$, $7$, and $8$.

\subsection{Nonexistence Theorems}

\begin{theorem}[Nonexistence of LCD Codes with Odd Length-Codimension Difference]\label{thm:nonexist}
Let $m \geq 7$ and $n \geq 2^{m-2} + 2^{m-4} - m + 1$.
Any LCD $[n,n-m,d]$ code with $d \geq 4$ must satisfy $n \equiv m \pmod{2}$ and $n \le 2^{m-1} - m$.
\end{theorem}

\begin{proof}
Let $C$ be an LCD $[n,n-m,d]$ code with $d \geq 4$. By \cref{prop:lcd-cap} and \cref{prop:lcd-nonsingular}, there exists a cap $S \subseteq \mathbb{F}_2^m$ with $\lvert S \rvert = n$ and $\langle S \rangle = \mathbb{F}_2^m$ such that $C = C_S$ and $U_S$ is nonsingular.
Since $U_S$ is nonsingular, we have $\rank(U_S) = m$, so that $\lvert S \rvert = n \geq 2^{m-2} + 2^{m-4} - \rank(U_S) + 1$.
Hence, by \cref{thm:main}, $S$ is contained in a maximal cap $T$ of the form $T = \mathbb{F}_2^m \setminus \langle v \rangle^\perp$ for some nonzero $v \in \mathbb{F}_2^m$.
\cref{lem:even-code} shows that $C_S$ is an even code.
Therefore, we have proved that $C_S$ is an even LCD code.
By \cref{lem:even-lcd}, the dimension $n-m$ must be even.
This implies $n \equiv m \pmod{2}$, as required.
The latter assertion is immediate from \cref{thm:main}.
\end{proof}

Put $n_0 = 2^{m-2} + 2^{m-4} - m + 1$. The upper bound $n \leq 2^{m-1} - m$ and the condition $n \geq n_0$ in \cref{thm:nonexist} are tight, as the following two constructions show.
First, \cref{prop:construction_S1} gives an LCD code with minimum distance $\geq 4$ of length $2^{m-1}-m$, so the upper bound is attained.
Second, \cref{prop:construction_S2} gives an LCD code with minimum distance $\geq 4$ of length $n_0 - 2 = 2^{m-2} + 2^{m-4} - (m+1)$. Since $n_0 - 2 \not\equiv m \pmod 2$, the conclusion of \cref{thm:nonexist} fails for $n = n_0 - 2$; that is, the condition $n \geq n_0$ cannot be relaxed to $n \geq n_0 - 2$.

\begin{proposition}\label{prop:construction_S1}
For $m \geq 7$, there exists an LCD $[2^{m-1}-m,\, 2^{m-1}-2m,\, \geq 4]$ code.
\end{proposition}
\begin{proof}
Let $H = \gen{e_1+e_2+\cdots +e_m}^{\perp} \subset \mathbb{F}_2^m$, and let $S_1 = \mathbb{F}_2^m \setminus (H\cup \{e_1,e_2,\ldots,e_m\})$.
Since $S_1 \subset \mathbb{F}_2^m \setminus H$, the set $S_1$ is a cap. 
Since $\mathbb{F}_2^m = H \sqcup S_1 \sqcup \{e_1,\ldots,e_m\}$, we have $U_{S_1}=U_{\mathbb{F}_2^m}+U_H+U_{\{e_1,\ldots, e_m\}}$ by \eqref{eq:additivity}. By \cref{lem:rank-bounds}(2), $U_{\mathbb{F}_2^m}=U_H=0$. Hence $U_{S_1}=U_{\{e_1,\ldots, e_m\}}=I_m$, which is nonsingular.
In particular, $\rank(H_{S_1}) \geq \rank(U_{S_1}) = m$ by \cref{lem:rank-bound}, so $\langle S_1 \rangle = \mathbb{F}_2^m$.
Therefore, $C_{S_1}$ is a $[2^{m-1}-m,2^{m-1}-2m,\geq 4]$ code by \cref{prop:cap-to-code}, and it is LCD by \cref{prop:lcd-nonsingular}.
(This construction is analogous to that in~\cite[Theorem~14]{characterizations2024}.)
\end{proof}

\begin{proposition}\label{prop:construction_S2}
For $m \geq 7$, there exists an LCD $[2^{m-2}+2^{m-4}-(m+1),\, 2^{m-2}+2^{m-4}-(2m+1),\, \geq 4]$ code.
\end{proposition}
\begin{proof}
Define $T, S_2 \subset \mathbb{F}_2^m$ by
    \begin{align*}
        T := &\{v=(a_1,a_2,\ldots,a_m)\in \mathbb{F}_2^m \mid v\not\in H \text{ and at most one of $a_1,a_2,a_3$ is $1$} \} \\ \cup\; &\{v=(a_1,a_2,\ldots,a_m)\in \mathbb{F}_2^m \mid v\in H \text{ and } a_1=a_2=a_3=1\}, \\
        S_2 := &\;T\setminus\{e_1,e_2,\ldots,e_m,e_1+e_2+e_3+e_4\},
    \end{align*}
where $H = \gen{e_1+e_2+\cdots +e_m}^{\perp}$.
Note that $\lvert T \rvert = 4 \cdot 2^{m-4} + 2^{m-4} = 2^{m-2} + 2^{m-4}$ and that the $m+1$ vectors $e_1,\ldots,e_m, e_1+e_2+e_3+e_4$ belong to $T$, so $\lvert S_2 \rvert = 2^{m-2}+2^{m-4}-(m+1)$.
Note that $0 \notin T$. We verify that $T$ is a cap by showing that for any distinct $u,v\in T$, we have $u+v\notin T$.
Write $u=(u_1,\ldots,u_m)$ and $v=(v_1,\ldots,v_m)$, and let $w=u+v$.

\emph{Case 1: $u,v\in T\setminus H$.}
Both $u$ and $v$ have odd weight, so $w$ has even weight and thus $w\in H$.
For $w$ to lie in $T\cap H$, we would need $w_1=w_2=w_3=1$.
Since $u\notin H$, at most one of $u_1,u_2,u_3$ equals $1$, and similarly for $v$.
If both have zero $1$'s among the first three coordinates, then $(w_1,w_2,w_3)=(0,0,0)$.
If one has zero and the other has exactly one, then $w$ has exactly one $1$.
If both have exactly one $1$, then $w$ has either zero or two $1$'s among the first three coordinates.
In every sub-case, $w$ does not have $w_1=w_2=w_3=1$, so $w\notin T$.

\emph{Case 2: $u\in T\setminus H$, $v\in T\cap H$.}
Here $u$ has odd weight and $v$ has even weight, so $w$ has odd weight and $w\notin H$.
For $w$ to lie in $T\setminus H$, we would need at most one of $w_1,w_2,w_3$ to equal $1$.
Since $v\in T\cap H$, we have $(v_1,v_2,v_3)=(1,1,1)$.
If $(u_1,u_2,u_3)=(0,0,0)$, then $(w_1,w_2,w_3)=(1,1,1)$.
If $u$ has exactly one $1$ among the first three coordinates, then $w$ has exactly two $1$'s.
In both sub-cases, $w$ has more than one $1$ among the first three coordinates, so $w\notin T$.

\emph{Case 3: $u,v\in T\cap H$.}
Both have even weight, so $w$ has even weight and $w\in H$.
Since $u_1=u_2=u_3=v_1=v_2=v_3=1$, we have $(w_1,w_2,w_3)=(0,0,0)$, so $w\notin T\cap H$.
Since $w\in H$, we also have $w\notin T\setminus H$.

\noindent In all cases $u+v\notin T$, so $T$ is a cap. Hence so is $S_2$.
Since $\Per(T) = \gen{e_4, e_5, \ldots, e_m} \cap H$ has dimension $m-4 \geq 3$, we have $U_T=0$ by \cref{lem:rank-bounds}(2). Therefore, by \eqref{eq:additivity},
    \begin{align*}
        U_{S_2} &= U_{\{e_1,\ldots, e_m\}}+U_{\{e_1+e_2+e_3+e_4\}}\\
        &=\left(
    \begin{array}{cccc}
0 & 1 & 1 & 1\\
1 & 0 & 1 & 1\\
1 & 1 & 0 & 1\\
1 & 1 & 1 & 0
    \end{array}
    \right)\oplus  I_{m-4},
    \end{align*}
where $\oplus$ denotes the direct sum (block diagonal) of matrices. Over $\mathbb{F}_2$, the determinant of the $4 \times 4$ block equals $-3 \equiv 1$, so $U_{S_2}$ is nonsingular.
In particular, $\rank(H_{S_2}) \geq \rank(U_{S_2}) = m$ by \cref{lem:rank-bound}, so $\langle S_2 \rangle = \mathbb{F}_2^m$.
Therefore, $C_{S_2}$ is a $[2^{m-2}+2^{m-4}-(m+1),\,2^{m-2}+2^{m-4}-(2m+1),\,\geq 4]$ code by \cref{prop:cap-to-code}, and it is LCD by \cref{prop:lcd-nonsingular}.
\end{proof}

\begin{corollary}[Nonexistence for Codimension $7$]\label{cor:nonexist_m7}
Let $n \geq 34$.
Any LCD $[n,n-7,d]$ code with $d \geq 4$ must satisfy $n \equiv 1 \pmod{2}$ and $n \le 57$.
\end{corollary}
\begin{proof}
    This is a specialization of \cref{thm:nonexist} for $m=7$.
\end{proof}

\cref{thm:nonexist} requires $m \ge 7$.
For $4 \le m \le 6$, we can use similar arguments to obtain the following.

\begin{theorem}\label{thm:nonexist_small}
Let $4 \le m \le 6$. Assume that either $n \geq 2^{m-2}+2$, or $m=6$ and $n=2^{m-2}+1$.
Any LCD $[n,n-m,d]$ code with $d \geq 4$ must satisfy $n \equiv m \pmod{2}$ and $n \le 2^{m-1} - m$.
\end{theorem}

\begin{proof}
We proceed in a similar manner to the proofs of \cref{thm:main} and~\cref{thm:nonexist}.
Let $C$ be an LCD $[n,n-m,d]$ code with $d \geq 4$. By \cref{prop:lcd-cap} and \cref{prop:lcd-nonsingular}, there exists a cap $S \subseteq \mathbb{F}_2^m$ with $\lvert S\rvert = n$ and $\langle S \rangle = \mathbb{F}_2^m$ such that $C = C_S$ and $U_S$ is nonsingular.

Let $T$ be a maximal cap containing $S$.
Since $n \ge 2^{m-2} + 1$, by \cref{thm:bw} there exists $j \in \{0,1,\dots,m-4\} \cup \{m-2\}$ such that $\lvert T\rvert =2^{m-2}+2^j$. 

First, we prove $\rank(U_T)\leq 2$. 
If $j\neq 0$, then $\Per(T) \neq \{0\}$ by \cref{thm:bw}, so this follows from \cref{lem:rank-bounds}(1).
If $j = 0$, then $n \leq \lvert T \rvert = 2^{m-2}+1$, so $n = 2^{m-2}+1$, which by assumption forces $m=6$. By \cite[Corollary 10.2(1)]{bruen1999long}, every maximal cap of size $2^{m-2}+1$ in $\mathbb{F}_2^m$ with $m \leq 6$ has a tangent hyperplane, that is, a hyperplane meeting the cap in exactly one point \cite[Definition~4.4]{bruen1999long}. Hence there exists a hyperplane $H\subset \mathbb{F}_2^m$ such that $\lvert T\cap H\rvert=1$. Hence $\rank(U_T) \leq 2$ by \cref{lem:rank-bounds}(3).

We aim to prove $j = m-2$. 
To this end, suppose to the contrary that $4\le m \le 6$ and $j \le m-4$.
With $X = T \setminus S$ as in the proof of \cref{thm:main}, we have
\begin{align*}
\rank(U_T)
&\ge \rank(U_S) - \lvert X\rvert\\
&= \rank(U_S) - (\lvert T \rvert - \lvert S\rvert)\\
&\ge m - ((2^{m-2} + 2^{m-4}) - (2^{m-2}+1))\\
&= (m+1) - 2^{m-4}.
\end{align*}
Combining these two inequalities, we obtain
\[
2 \ge \rank(U_T) \ge (m+1) - 2^{m-4}.
\]
One can easily verify that for $4 \le m \le 6$, the above inequality cannot be satisfied.
Therefore, we must have $j = m-2$, as required.
This means $T = \mathbb{F}_2^m \setminus \langle v \rangle^\perp$ for some nonzero $v \in \mathbb{F}_2^m$ by \cref{thm:bw}.
This also gives $n \le 2^{m-1} - m$ by the same argument as in the proof of \cref{thm:main}. (Note that since $m \geq 4$, we have $\lvert \Per(T)\rvert = 2^{m-1} \geq 8$, so $U_T=0$ by \cref{lem:rank-bounds}(2).)
Finally, by the same argument as in the proof of \cref{thm:nonexist}, we see that $n \equiv m \pmod 2$.
This completes the proof.
\end{proof}

\begin{corollary}[Nonexistence for Codimension $6$]\label{cor:nonexist_m6}
Let $n \geq 17$.
Any LCD $[n,n-6,d]$ code with $d \geq 4$ must satisfy $n \equiv 0 \pmod{2}$ and $n \le 26$.
\end{corollary}
\begin{proof}
    This is a specialization of \cref{thm:nonexist_small} for $m=6$. 
\end{proof}

\subsection{Complete Results for Codimensions $6$ and $7$}

Previously, the values of $d_2^E(n,n-6)$ were completely determined in~\cite{characterizations2024}. For odd $n$ with $17 \le n \le 26$, however, these values rely on computer calculations: the cases $n \le 24$ were settled in~\cite[Table~15]{araya2020minimum}, and the nonexistence of an LCD $[25,19,4]$ code was shown by an exhaustive computer search in~\cite[Section~4.1]{bouyuklieva2021optimal}.
For $d_2^E(n,n-7)$, all values were determined in~\cite[Theorem~19]{characterizations2024} except for $n \in \{34,36,\ldots,56\} \cup \{58,59,\ldots,64\}$, where it was only shown that $d_2^E(n,n-7) \in \{3,4\}$.

\cref{cor:nonexist_m6} shows that LCD $[n,n-6,4]$ codes do not exist for odd $n$ with $17 \le n \le 26$, providing a computation-free proof for the nonexistence results that previously required exhaustive search.
\cref{cor:nonexist_m7} rules out the existence of LCD $[n,n-7,4]$ codes for all even $n$ with $34 \le n \le 57$, as well as for all $n \ge 58$.
Hence $d_2^E(n,n-7) = 3$ for these $n$, which together with~\cite[Theorem~19]{characterizations2024} completely determines $d_2^E(n,n-7)$.

\begin{theorem}[Complete Determination for Codimension $7$]\label{thm:complete-m7}
Let $d_2^E(n,n-7)$ be the largest minimum distance among binary
Euclidean LCD $[n,n-7]$ codes.
Then 
\[
d_2^E(n,n-7)
=
\begin{cases}
7& \text{ if } n=8, \\
6& \text{ if } n=9, \\
5& \text{ if } n=10, \\
4& \text{ if } n \in 
\{11,12,\ldots,33\} \cup
\left\{
\begin{array}{ll}
35,37,39,41,\\
43,45,47,49,\\
51,53,55,57
\end{array}
\right\},\\
%4& \text{ if } n \in \{11,12,\ldots,31\} \\ 
%& \qquad \cup \{2m+1 \mid m \in \{16,17,\ldots,28\}\},\\
3& \text{ if } n \in 
\left\{
\begin{array}{ll}
34,36,38,40,\\
42,44,46,48,\\
50,52,54,56
\end{array}
\right\}
\cup \{58,59,\ldots,120\},\\
2& \text{ if } n \in \{121,122,\ldots\}.\\
\end{cases}
\]
\end{theorem}

\subsection{Complete Results for Codimension $8$}

We now completely determine the largest minimum weights $d_2^E(n,n-8)$ among binary LCD $[n,n-8]$ codes for arbitrary~$n$.

It is known~\cite[Table~1]{galvez2018bounds}, \cite[Table~3]{HaradaSaito}, \cite[Tables~1 and~2]{bouyuklieva2021optimal} that
\begin{equation}\label{eq:known-m8-small}
d_2^E(n,n-8)
=
\begin{cases}
9 & \text{if } n=9,\\
6 & \text{if } n=10, \\
5 & \text{if } n \in \{11,12,\ldots,17\},\\
4 & \text{if } n \in \{18,19,\ldots,40\}.
\end{cases}
\end{equation}
By Theorems~11 and~14 in~\cite{characterizations2024}, we have that
\begin{equation}\label{eq:known-m8-large}
d_2^E(n,n-8) =
\begin{cases}
3 & \text{if } n = 247,\\
2 & \text{if } n \in \{248, 249, \ldots\}.
\end{cases}
\end{equation}
It remains to determine the values $d_2^E(n,n-8)$ for
\[
n \in \{41, 42, \ldots, 246\}.
\]

First, we collect upper bounds on $d_2^E(n,n-8)$.
Let $d(n,k)$ denote the largest minimum distance among all (not necessarily LCD) binary $[n,k]$ codes.
It is known~\cite{codetables} that
\begin{equation}\label{eq:d-m8}
d(n,n-8)
=
\begin{cases}
4 & \text{if } n \in \{41,42,\ldots,128\},\\
3 & \text{if } n \in \{129,130,\ldots,247\}.
\end{cases}
\end{equation}
Furthermore, the main theorem of this paper gives the following nonexistence result, which is essential for the determination of $d_2^E(n,n-8)$.
\begin{corollary}[Nonexistence for Codimension $8$]\label{cor:nonexist_m8}
Let $n \geq 73$.
Any LCD $[n,n-8,d]$ code with $d \geq 4$ must satisfy $n \equiv 0 \pmod{2}$ and $n \le 120$.
\end{corollary}
\begin{proof}
    This is a specialization of \cref{thm:nonexist} for $m=8$.
\end{proof}

Now, we complete the determination by constructing LCD $[n,n-8]$ codes that meet the above upper bounds.
Recall from \cref{prop:lcd-nonsingular} that for $S \subseteq \mathbb{F}_2^m \setminus \{0\}$ with $\langle S \rangle = \mathbb{F}_2^m$, the code $C_S$ with parity-check matrix $H_S$ is an LCD $[\lvert S \rvert, \lvert S \rvert-m]$ code if and only if $U_S$ is nonsingular.
Moreover, $C_S$ has minimum distance at least $3$, since the columns of $H_S$ are distinct and nonzero, and it has minimum distance at least $4$ if and only if $S$ is a cap (\cref{prop:lcd-cap} and~\cref{prop:cap-to-code}).
We performed a computational search for subsets $S \subseteq \mathbb{F}_2^8 \setminus \{0\}$ satisfying these conditions.
\begin{itemize}
    \item For $n \in \{41,42,\ldots,72\} \cup \{74,76,\ldots,120\}$, we found caps $S$ of size $n$ in $\mathbb{F}_2^8$ with $U_S$ nonsingular, yielding LCD $[n, n-8, \geq 4]$ codes $C_S$.
    \item For odd $n \in \{73, 75, \ldots, 119\}$ and $n \in \{121, 122, \ldots, 246\}$, we found subsets $S \subseteq \mathbb{F}_2^8 \setminus \{0\}$ of size $n$ with $U_S$ nonsingular, yielding LCD $[n, n-8, \geq 3]$ codes $C_S$.
\end{itemize}
The sets $S$ are available at~\cite{github_repo_codim8}. 
These constructions show that
\begin{equation}\label{eq:lcd-cap-m8}
d_2^E(n,n-8) \ge 4 \quad \text{if } n \in \{41,42,\ldots,72\} \cup \{74,76,\ldots,120\},
\end{equation}
\begin{equation}\label{eq:lcd-d3-m8}
d_2^E(n,n-8) \ge 3 \quad \text{if } n \in \{73,75,\ldots,119\} \cup \{121,122,\ldots,246\}.
\end{equation}
\cref{tab:m8_bounds} summarizes how the upper and lower bounds combine to determine $d_2^E(n,n-8)$ for each range of $n$.
From \eqref{eq:known-m8-small}--\eqref{eq:lcd-d3-m8}, we have the following theorem.
\begin{theorem}[Complete Determination for Codimension $8$]\label{thm:complete-m8}
Let $d_2^E(n,n-8)$ be the largest minimum distance among binary
Euclidean LCD $[n,n-8]$ codes.
Then
\[
d_2^E(n,n-8)
=
\begin{cases}
9& \text{ if } n=9, \\
6& \text{ if } n=10, \\
5& \text{ if } n \in \{11,12,\ldots,17\}, \\
4& \text{ if } n \in \{18,19,\ldots,72\}
\cup \{74,76,\ldots,120\},\\
3& \text{ if } n \in
\{73,75,\ldots,119\}
\cup \{121,122,\ldots,247\},\\
2& \text{ if } n \in \{248,249,\ldots\}.\\
\end{cases}
\]
\end{theorem}

\begin{table}[htb]
\centering
\caption{Determination of $d_2^E(n,n-8)$: sources of upper and lower bounds.}\label{tab:m8_bounds}
\begin{tabular}{ccll}
\hline
$n$ & $d_2^E$ & Lower bound & Upper bound \\
\hline
$9$--$40$ & $9$--$4$ & \cite{galvez2018bounds,HaradaSaito,bouyuklieva2021optimal} & \cite{galvez2018bounds,HaradaSaito,bouyuklieva2021optimal} \\
$41$--$72$ & $4$ & construction~\cite{github_repo_codim8} & \cite{codetables} ($d(n,n-8) \le 4$) \\
even $74$--$120$ & $4$ & construction~\cite{github_repo_codim8} & \cite{codetables} ($d(n,n-8) \le 4$) \\
odd $73$--$119$ & $3$ & construction~\cite{github_repo_codim8} & \cref{cor:nonexist_m8} ($n$ odd $\Rightarrow d \le 3$) \\
$121$--$128$ & $3$ & construction~\cite{github_repo_codim8} & \cref{cor:nonexist_m8} ($n > 120 \Rightarrow d \le 3$) \\
$129$--$246$ & $3$ & construction~\cite{github_repo_codim8} & \cite{codetables} ($d(n,n-8) \le 3$) \\
$247$ & $3$ & \cite{characterizations2024} & \cite{characterizations2024} \\
$\ge 248$ & $2$ & \cite{characterizations2024} & \cite{characterizations2024} \\
\hline
\end{tabular}
\end{table}

\section{Ternary and Quaternary Hermitian LCD Codes}\label{sec:nonbinary}

By \cref{prop:lcd-cap,prop:cap-to-code,prop:lcd-nonsingular}, which hold for $q \in \{2,3,4\}$, an LCD $[n,n-m,d]$ code over $\mathbb{F}_q$ with $d \geq 4$ exists if and only if there is a cap $S \subseteq \mathbb{F}_q^m$ with $\lvert S \rvert = n$, $\langle S \rangle = \mathbb{F}_q^m$ and nonsingular $U_S$.
The existence question for such codes is therefore a question about caps in $\PG(m-1,q)$ and the ranks of their Gram matrices.
In the binary case, \cref{thm:bw} describes all large caps, and this is what makes \cref{thm:main,thm:nonexist,thm:nonexist_small} possible.
In this section, we consider $q \in \{3,4\}$.
In \cref{sec:nb-ternary,sec:nb-quaternary}, we use known results on large caps in $\PG(5,3)$ and $\PG(4,4)$ to settle cases left open in~\cite{characterizations2024}.
In \cref{sec:nb-limits}, we explain why the phenomena of \cref{thm:nonexist,thm:nonexist_small} have no counterpart for $q \in \{3,4\}$.

\subsection{Ternary LCD Codes of Codimension $6$}\label{sec:nb-ternary}

The values $d_3^E(n,n-6)$ were determined in~\cite[Theorem~25]{characterizations2024} except for $51 \leq n \leq 56$, where it was shown that $d_3^E(n,n-6) \in \{3,4\}$.

\begin{proposition}\label{prop:ternary}
We have $d_3^E(n,n-6) = 3$ for $51 \leq n \leq 56$.
\end{proposition}

\begin{proof}
Suppose that there exists a ternary LCD $[n,n-6,4]$ code with $51 \leq n \leq 56$.
By \cref{prop:lcd-cap,prop:lcd-nonsingular}, it is of the form $C_S$ for some cap $S \subseteq \mathbb{F}_3^6$ with $\lvert S \rvert = n$ such that $U_S$ is nonsingular.
Every cap of size at least $49$ in $\PG(5,3)$ is contained in a $56$-cap~\cite{barat2004complete}, and the $56$-cap in $\PG(5,3)$ is unique up to projective equivalence~\cite{hill1978caps}.
Let $T$ be a $56$-cap containing $S$.
Let $T_0$ be the representative of the $56$-cap given in~\cite[Section~VII]{characterizations2024} (see also~\cite{hill1978caps}).
A direct computation shows that $U_{T_0} = 0$.
Since $T$ is projectively equivalent to $T_0$, we have $U_T = A^T U_{T_0} A = 0$ for some invertible matrix $A$ by~\eqref{eq:equivalence} and the remark following it.
Put $X = T \setminus S$. By~\eqref{eq:additivity} and \cref{lem:rank-bound}, we have $U_S = -U_X$ and
\[
\rank(U_S) = \rank(U_X) \leq \lvert X \rvert = 56 - n \leq 5,
\]
which contradicts the nonsingularity of $U_S$.
\end{proof}

Together with~\cite[Theorem~25]{characterizations2024}, \cref{prop:ternary} completes the determination of $d_3^E(n,n-6)$.

\begin{corollary}[Complete Determination of $d_3^E(n,n-6)$]\label{cor:ternary}
\[
d_3^E(n,n-6)
=
\begin{cases}
7 & \text{if } n = 7,\\
5 & \text{if } n \in \{8,9,\ldots,14\},\\
4 & \text{if } n \in \{15,16,\ldots,50\},\\
3 & \text{if } n \in \{51,52,\ldots,358\},\\
2 & \text{if } n \in \{359,360,\ldots\}.
\end{cases}
\]
\end{corollary}

\subsection{Quaternary Hermitian LCD Codes of Codimension $5$}\label{sec:nb-quaternary}

The values $d_4^H(n,n-5)$ were determined in~\cite[Theorem~32]{characterizations2024} except for $37 \leq n \leq 41$, where it was shown that $d_4^H(n,n-5) \in \{3,4\}$.

\begin{proposition}\label{prop:quaternary}
We have $d_4^H(37,32) = 4$.
\end{proposition}

\begin{proof}
Let $S \subseteq \mathbb{F}_4^5$ be the set of columns of the following matrix, where we write $2 = \omega$ and $3 = \omega^2$:
\[
H_S = \left(
\begin{array}{c}
\mathtt{0213010223333122103103230321021023032}\\
\mathtt{0132101013221322010121332022301101303}\\
\mathtt{0303223220123321330101023302112102012}\\
\mathtt{0032111103331223101030223133210010212}\\
\mathtt{1130331132032231021013303320332120102}
\end{array}
\right).
\]
The set $S$ is obtained from a $41$-cap $T$ given in~\cite{edel1999cap} by deleting the four points $e_1, e_2, e_3, e_4$; in particular, $S$ is a cap.
A direct computation shows that $\langle S \rangle = \mathbb{F}_4^5$ and
\[
U_S = \left(
\begin{array}{ccccc}
0 & \omega^2 & 1 & 1 & 1\\
\omega & 0 & \omega & \omega & \omega\\
1 & \omega^2 & 0 & 1 & 1\\
1 & \omega^2 & 1 & 0 & 1\\
1 & \omega^2 & 1 & 1 & 1
\end{array}
\right),
\qquad \det U_S = 1.
\]
By \cref{prop:cap-to-code,prop:lcd-nonsingular}, $C_S$ is a quaternary Hermitian LCD $[37,32]$ code with minimum distance at least $4$. Since $d_4^H(37,32) \leq 4$, the result follows.
\end{proof}

The values $d_4^H(n,n-5)$ for $38 \leq n \leq 41$ remain open.

\subsection{Differences from the Binary Case}\label{sec:nb-limits}

The proofs of \cref{thm:main,thm:nonexist,thm:nonexist_small} combine the rank bound with three ingredients which are specific to the binary case.
\begin{enumerate}[label=\textup{(\alph*)}]
\item \emph{Classification of large caps.} By \cref{thm:bw}, every sufficiently large cap in $\PG(m-1,2)$ is contained in the complement of a hyperplane. For $q \geq 3$, the complement of a hyperplane is not even a cap, since every line which is not contained in the hyperplane has $q \geq 3$ points outside it, and no classification of large caps in $\PG(m-1,q)$ is known in general.
\item \emph{Periods.} \cref{lem:rank-bounds}(1) and (2) use the translations $x \mapsto x + w$ of $\mathbb{F}_2^m$, which act on the points of $\PG(m-1,2)$ only because nonzero vectors and projective points coincide over $\mathbb{F}_2$.
\item \emph{Parity.} The condition $n \equiv m \pmod 2$ relies on \cref{lem:even-lcd}. This lemma has no counterpart for ternary codes, and quaternary even codes are Hermitian self-orthogonal~\cite[Theorem~1.4.10]{huffman2003fundamentals}, hence never Hermitian LCD.
\end{enumerate}
Accordingly, the values determined in~\cite{characterizations2024} and in this section do not exhibit the phenomena of \cref{thm:nonexist,thm:nonexist_small}.
First, there is no parity phenomenon: $d_3^E(n,n-5) = 4$ for all $7 \leq n \leq 19$, $d_3^E(n,n-6) = 4$ for all $15 \leq n \leq 50$, and $d_4^H(n,n-5) = 4$ for all $12 \leq n \leq 36$~\cite[Theorems~22, 25 and~32]{characterizations2024}, regardless of the parity of $n$.
Second, the bound $n \leq 2^{m-1} - m$ in \cref{thm:nonexist,thm:nonexist_small} states that the length is at most the maximum size of a cap minus $m$, and the analogue of this statement fails: the largest cap in $\PG(4,3)$ has size $20$ (equivalently, ternary $[n,n-5,4]$ codes exist if and only if $n \leq 20$~\cite{codetables}), so the analogue would read $n \leq 20 - 5 = 15$, whereas ternary LCD $[n,n-5,4]$ codes exist for $16 \leq n \leq 19$~\cite[Theorem~22]{characterizations2024}; likewise, the largest cap in $\PG(4,4)$ has size $41$~\cite{edel1999cap}, whereas $d_4^H(37,32) = 4$ by \cref{prop:quaternary}.

For $q \in \{3,4\}$, the determination of $d_3^E(n,n-m)$ and $d_4^H(n,n-m)$ in the range where they are at least $4$ thus remains a problem about large caps in $\PG(m-1,q)$: one has to know which large caps exist and what the ranks of their Gram matrices are.
\Cref{prop:ternary,prop:quaternary} show what can be done when this information is available, and the open values $d_4^H(n,n-5)$ for $38 \leq n \leq 41$ show where it is missing.

\section{Conclusion}\label{sec:conc}

We have established a connection between LCD codes and caps in projective space through the Gram matrix $U_S = H_S \overline{H_S}^T$: a code $C_S$ from a cap is LCD if and only if $U_S$ is nonsingular.
Combined with Bruen and Wehlau's structure theory of large maximal caps, this framework yields nonexistence theorems for LCD codes with minimum distance at least $4$.
For codimension $6$, we provide a computation-free proof for results that previously required exhaustive search.
For codimensions $7$ and $8$, we completely resolve the open cases, determining $d_2^E(n,n-7)$ and $d_2^E(n,n-8)$ for all $n$.
To our knowledge, this is the first theoretical framework that explains the alternating pattern in $d_2^E(n,n-m)$ and applies uniformly to all $m \ge 4$.
For ternary LCD codes and quaternary Hermitian LCD codes, the same framework combined with known results on large caps determines $d_3^E(n,n-6)$ for all $n$ and $d_4^H(37,32)$.

As discussed in \cref{sec:nonbinary}, this connection holds for $q \in \{2,3,4\}$ and reduces the existence of ternary LCD codes and quaternary Hermitian LCD codes with minimum distance at least $4$ to a question about large caps in $\PG(m-1,q)$. Known results on large caps settle a few values left open in~\cite{characterizations2024}, but the phenomena behind our theorems have no counterpart for $q \in \{3,4\}$, and further progress for these codes requires further knowledge of large caps.
Other open problems include determining $d_2^E(n,n-m)$ completely for $m \geq 9$ and classifying caps $S \subseteq \mathbb{F}_2^m$ near the threshold with nonsingular $U_S$.

\section*{Acknowledgments}
The authors thank the Associate Editor and the anonymous reviewers for their careful reading and valuable comments, which helped to improve the manuscript.
This work was supported by JSPS KAKENHI Grant Number JP25K17290.

\bibliographystyle{plain}

\end{document}